\documentclass[journal, letterpaper, twocolumn]{IEEEtran}
\usepackage{citesort}
\usepackage{color}
\usepackage[latin9]{inputenc}
\usepackage{amsthm}
\usepackage{amsmath}
\usepackage{amsthm}
\usepackage{amssymb}
\usepackage{graphicx}
\usepackage{algorithmic}
\usepackage{algorithm}
\usepackage{mathrsfs}

\newtheorem{theorem}{Theorem}
\newtheorem{lemma}{Lemma}
\newtheorem{proposition}[theorem]{Proposition}
\theoremstyle{definition}
\newtheorem{definition}{Definition}

\theoremstyle{remark}
\newtheorem{remark}{Remark}

\newcommand{\PU}{\ensuremath {\mathit{PU}}}
\newcommand{\n}{\ensuremath {\mathit{n}}}
\newcommand{\SU}{\ensuremath {\mathit{SU}}}

\begin{document}

\title{Efficient Spectrum Availability Information Recovery for Wideband DSA Networks: A Weighted Compressive Sampling Approach}

\author{Bassem~Khalfi,~\IEEEmembership{Student~Member,~IEEE,}
        Bechir~Hamdaoui,~\IEEEmembership{Senior~Member,~IEEE,}
        Mohsen~Guizani,~\IEEEmembership{Fellow,~IEEE,}
        and~Nizar~Zorba,~\IEEEmembership{Senior~Member,~IEEE}

\thanks{Bassem Khalfi and Bechir Hamdaoui are with Oregon State University, Corvallis, OR, USA (e-mail: khalfib@eecs.orst.org).}
\thanks{Mohsen Guizani and Nizar Zorba are with Qatar University, Doha, Qatar.}
\thanks{This work was supported in part by the US National Science Foundation (NSF) under NSF award CNS-1162296.}}

\maketitle
\begin{abstract}
Compressive sampling has great potential for making wideband spectrum sensing possible at sub-Nyquist sampling rates. As a result, there have recently been research efforts that leverage compressive sampling to enable efficient wideband spectrum sensing. These efforts consider homogenous wideband spectrum, where all bands are assumed to have similar PU traffic characteristics. In practice, however, wideband spectrum is not homogeneous, in that different spectrum bands could present different PU occupancy patterns. In fact, the nature of spectrum assignment, in which applications of similar types are often assigned bands within the same block, dictates that wideband spectrum is indeed heterogeneous.
In this paper, we consider heterogeneous wideband spectrum, and exploit its inherent, block-like structure to design efficient compressive spectrum sensing techniques that are well suited for heterogeneous wideband spectrum. We propose a weighted $\ell_1-$minimization sensing information recovery algorithm that achieves more stable recovery than that achieved by existing approaches while accounting for the variations of spectrum occupancy across both the time and frequency dimensions. In addition, we show that our proposed algorithm requires a lesser number of sensing measurements when compared to the state-of-the-art approaches.
\end{abstract}


\section{Introduction}
\label{sec:intro}
Spectrum sensing is a key component of cognitive radio networks (CRNs), essential for enabling dynamic and opportunistic spectrum access~\cite{akyildiz2011,paisana2014radar}.
It essentially allows secondary users (\SU s) to know whether and when a licensed band is available prior to using it so as to avoid harming primary users (\PU s).
Due to its vital role, over the last decade or so, a tremendous amount of research has focused on developing techniques and approaches that enable efficient spectrum sensing~\cite{axell2012spectrum,patil2016survey}. Most of the focus has, however, been on single-band spectrum sensing, and the focus on wideband spectrum sensing has recently received increased attention~\cite{sun2013wideband}.

The key advantage of wideband spectrum sensing over its single-band counterpart is that it allows \SU s to locate spectrum opportunities in wider ranges of frequencies by performing spectrum sensing across multiple bands at the same time. Being able to perform wideband spectrum sensing is becoming a crucial requirement of next-generation CRNs, especially with the emergence of IoT and 5G technologies~\cite{al2016information,niu2015survey,gohil20135g}. This wideband spectrum sensing requirement is becoming even more stringent with FCC's recent new rules for opening up millimeter wave bands' use for wireless broadband devices in frequencies above 24 GHz~\cite{FCC-mmW-16}.

The challenge, however, with wideband spectrum sensing is that it requires high sampling rates, which can incur significant sensing overhead in terms of energy, computation, and communication.
Motivated by the sparsity nature of spectrum occupancy~\cite{chen2014survey} and in an effort to address the overhead caused by these high sampling rates, researchers have focused on exploiting compressive sampling to make wideband spectrum sensing possible at sub-Nyquist sampling rates (e.g.~\cite{sharma2016application,Qin2016TSP,Mishali2010JSTSP,Needell2010JSTSP,qin2016data}).

These research efforts have focused mainly on {\em homogenous} wideband spectrum, meaning that the entire wideband spectrum is considered as one single block with multiple bands, and the sparsity level is estimated across all bands and considered to be the same for the entire wideband spectrum.
However, in spectrum assignment, applications of similar types (TV, satellite, cellular, etc.) are often assigned bands within the same band block, and different application types exhibit different traffic occupancy patterns and behaviors. This suggests that wideband spectrum is block-like {\em heterogeneous}, in that band occupancy patterns are not the same across the different band blocks.
Therefore, sparsity levels may vary significantly from one block to another, a trend that has also been confirmed by recent measurement studies~\cite{chen2014survey,yilmaz2016determination}.

In this paper, we exploit this inherent, block-like structure of wideband spectrum to design efficient compressive spectrum sensing techniques that are well suited for {\em heterogeneous} wideband spectrum access in {\em noisy} wireless environments. To the best of our knowledge, this is the first work that exploits this spectrum occupancy heterogeneity inherent to wideband spectrum to develop efficient compressive sensing techniques. Specifically, we propose a wideband sensing information recovery algorithm that is more stable and robust than existing approaches. The proposed technique accounts for spectrum occupancy variations across both time and frequency, and requires a lesser number of sensing measurements when compared to the state-of-the-art approaches.

\subsection{Related Work}
The work of Tian et al.~\cite{tian2007compressed} is the first to use compressive sampling for wideband spectrum sensing. Since then, a lot of work has exploited compressive sampling theory to enable wideband sensing at sub-Nyquist sampling rates~\cite{tian2012cyclic,Qin2016TSP,Mishali2010JSTSP,Needell2010JSTSP,tropp2010beyond,qin2016data,sun2016cooperative}. A common factor among these works is the assumption that the sparsity level is fixed over time.
In an effort to relax this assumption, the authors in~\cite{wang2012sparsity} propose a two-step algorithm, where at each sensing period, the sparsity level is first measured, and then used to adjust the total number of measurements. The issue, however, with this approach lies in its computational complexity. An autonomous compressive spectrum sensing algorithm is proposed in~\cite{jiang2016achieving}  that does not require the knowledge of the instantaneous sparsity level. However, this technique still assumes that the sparsity level is bounded and also \PU's signal is wide-sense stationary which is not usually guaranteed in practice. Cooperative wideband spectrum sensing is also considered in~\cite{sun2016cooperative} where a multi-rate sub-Nyquist recovery approach is proposed and analyzed under diverse fading channels.

There have also been some research efforts that aim at exploiting additional knowledge about the signal to further improve the sensing information recovery~\cite{vaswani2010modified,friedlander2012recovering,needell2016weighted,xu2010breaking,candes2008enhancing,khajehnejad2011analyzing,ahsen2015error}.
For instance,~\cite{vaswani2010modified} proposes a $\ell_1-$minimization based approach that exploits knowledge about the support\footnote{The support corresponds to the signal components that are non-zero.} of the sparse signal to recover information from noise-free measurements.
The authors in~\cite{friedlander2012recovering}
also exploit signal support information, but for recovering signals with noisy measurements. Their technique assigns a weight less than one to each index of the estimate of the support and one to all other indexes.
They show that this recovery approach is more stable and robust than standard $\ell_1-$minimization approaches when $50\%$ of the support is estimated correctly. This approach has been generalized for multiple weights in~\cite{needell2016weighted}, addressing the case where the support is estimated with different confidence levels.
These approaches, however, work well in applications where the support does not change much over time, like real-time dynamic MRI~\cite{vaswani2010modified} and video/audio decoding~\cite{friedlander2012recovering,needell2016weighted} applications. In the wideband spectrum sensing case where the signal support changes over time, an estimate of the support is too difficult to acquire in advance, making these approaches unsuitable.

There have also been attempts that exploit block sparsity information in signals to further improve signal recovery, though not in the context of wideband spectrum sensing~\cite{khajehnejad2011analyzing,ahsen2015error}.
For instance, the authors in~\cite{khajehnejad2011analyzing} consider noise-free measurements where the signal support is divided into two different subclasses with different sparsity levels. The focus of this work is on deriving the optimal weights that lead to the best recovery.
Also, in~\cite{ahsen2015error}, the authors study compressive sampling schemes for signals that only a few of their blocks are dense, whereas the rest of the blocks are zeros.

Unlike these previous works and as motivated by the real nature of wideband spectrum sparsity structure, our proposed framework considers time-varying and heterogeneous wideband spectrum occupancy.
We exploit this fine-grained sparsity structure to propose, which to the best of our knowledge, the first spectrum sensing information recovery scheme for {\em heterogeneous} wideband spectrum sensing with {\em noisy} measurements.

\subsection{Our Key Contributions}
In this paper, we make the following contributions:

\begin{itemize}

\item
We propose a weighted $\ell_1-$minimization algorithm that exploits the block-like, sparsity structure of the heterogeneous wideband spectrum to provide an efficient recovery of spectrum occupancy information in noisy CRN environments.
We design the weights of the algorithm in a way that spectrum blocks that are more likely to be occupied are favored during the search, thereby increasing the recovery performance.

\item We prove that our recovery algorithm outperforms existing approaches in terms of stability and robustness, and reduces sensing overhead by requiring lesser numbers of measurements.
     It does so while accounting for spectrum occupancy variations across both time and frequency.

\item We derive lower bounds on the probability of spectrum occupation, and use them to determine the sparsity levels that lead to further reduction in the sensing overhead.
\end{itemize}

It is important to mention that our proposed weighted compressive sampling framework, including the derived theoretical results, is not restricted to wideband spectrum sensing applications only. It can be applied to any other application where the signal to be recovered possesses block-like sparsity structure. We are hoping that this work
can be found useful for finding efficient solution methodologies to problems (with similar characteristics) in other disciplines and domains.

\subsection{Roadmap}
The remainder of the paper is structured as follows. In Section~\ref{sec:system_model}, we present our system model and the PU bands' occupancy model. Next, our proposed approach along with its performance analysis are presented in Section~\ref{sec:proposed}. The numerical evaluations are then presented in Section~\ref{sec:numerical_results}. Finally, our conclusions are given in Section~\ref{sec:conclusion}.

\section{Wideband Spectrum Sensing Model}
\label{sec:system_model}
In this section, we begin by presenting the studied heterogeneous wideband spectrum model. Then, we present the spectrum sensing preliminaries and setup.

\subsection{Wideband Occupancy Model}
We consider a heterogeneous wideband spectrum access system containing \n~frequency bands as illustrated by Fig.~\ref{fig:band_ocup}(a).
We assume that wideband spectrum accommodates multiple different types of user applications, where applications of the same type are allocated frequency bands within the same block. Therefore, we consider that wideband spectrum has a block-like occupation structure, where each block (accommodating applications of similar type) has different occupancy behavioral characteristics. The wideband spectrum can then be grouped into $g$ disjoint contiguous blocks, $\mathcal{G}_i, i=1,...,g$, 
with $\mathcal{G}_i\bigcap\mathcal{G}_j=\emptyset$ for $i\neq j$. Each block, $\mathcal{G}_i$, is a set of $n_i$ contiguous bands.
Like previous works~\cite{Sun2014TSP},
the state of each band $i$, $\mathcal{H}_i$, is modelled as $\mathcal{H}_i \thicksim \text{Bernoulli}(p_i)$
with parameter $p_i\in[0,1]$ ($p_i$ is the probability that band $i$ is occupied by a \PU). Assuming that the bands' occupancies within a block are independent from one another, then the average number of occupied bands is $\bar{k}_j
= \sum_{i\in \mathcal{G}_j } p_i$ for $j=1,...,g$.

Recall that one of the things that distinguishes this work from others is the fact that we consider a {\em heterogeneous} wideband spectrum; formally, this means that the average number $\bar{k}_j$ of the occupied bands in block $j$ can vary significantly from one block to another. The average occupancies, however, of the different bands within a given block are close to one another; i.e., $p_i\approx p_j$ for all $i,j\in \mathcal{G}_j$. Our proposed framework exploits such a block-like occupancy structure stemming from the wideband spectrum heterogeneity to design efficient compressive wideband spectrum sensing techniques.
For this, we assume that the blocks have sufficient different average sparsity levels (otherwise, blocks with similar sparsity levels are merged into one block with a sparsity level corresponding to their average). This is supported by practical observations where typically each block of bands is assigned to a particular application, and the average occupancy could be quite different from one block to another~\cite{yilmaz2016determination,mehdawi2015spectrum,harrold2011long}.
These averages are often available via measurement studies, and can easily be estimated, or provided by spectrum operators~\cite{mehdawi2015spectrum}.

\begin{figure}[t]
\centering{
\includegraphics[width=.9\columnwidth]{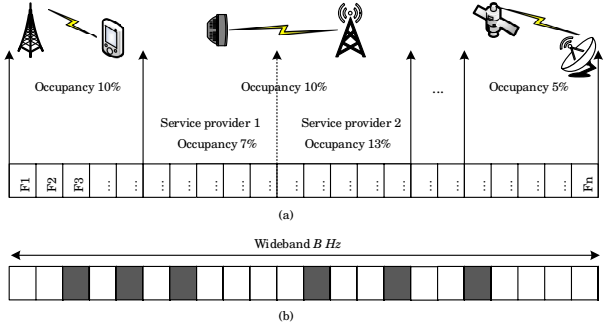}
\caption{$n$ frequency bands occupied by heterogeneous applications with different occupancy rates. The grey bands are occupied by primary users while the white bands are vacant. (a) is the statistical allocation while (b) is a realization of allocation in a given region at a given time slot. }
\label{fig:band_ocup}}
\end{figure}

\subsection{Secondary System Model}
We consider a SU performing the sensing of the entire wideband spectrum as illustrated by Fig.~\ref{fig:sys_mod}.
The time-domain signal $\boldsymbol{r}(t)$ received by the \SU~can be expressed as
\begin{equation}\label{eqn:signaltime}\nonumber
\boldsymbol{r}(t)=\boldsymbol{h}(t)\otimes\boldsymbol{s}(t)+\boldsymbol{w}(t),
\end{equation}
where $\boldsymbol{h}(t)$ is the channel impulse between the primary transmitters and the SU, $\boldsymbol{s}(t)$ is the PUs' signal, $\boldsymbol{w}(t)$ is an additive white Gaussian noise with mean $0$ and variance $\sigma^2$, and $\otimes$ is the convolution operator.
Ideally, we should take samples with at least twice the maximum frequency, $f_{\max}$, of the signal in order to recover the signal successfully. Let the sensing window be $[0,mT_0]$ with $T_0=1/(2f_{\max})$.
Assuming a normalized number of wideband Nyquist samples per band, then the vector of the taken samples is $\boldsymbol{r}(t)=[r(0),...,r((m_0-1)T_0)]^T$ where $r(i)=r(t)|_{t=iT_0}$ and $m_0=n$.
Note that a reasonable assumption that we make is that the sensing window length is assumed to be sufficiently small when compared to the time it takes a band state to change. That is, each band's occupancy is assumed to remain constant during each sensing time window.

To reveal which bands are occupied, we perform a discrete Fourier transform of the received signal $\boldsymbol{r}(t)$; i.e.,
\begin{equation}\label{eqn:signalfrequency} \nonumber
\boldsymbol{r}_f=\boldsymbol{h}_f\boldsymbol{s}_f+\boldsymbol{w}_f=\boldsymbol{x}+\boldsymbol{w}_f,
\end{equation}
where $\boldsymbol{h}_f$, $\boldsymbol{s}_f$, and $\boldsymbol{w}_f$ are the Fourier transforms of $\boldsymbol{h}(t)$, $\boldsymbol{s}(t)$, and $\boldsymbol{w}(t)$, respectively.
\begin{figure}
\centering{
\includegraphics[width=.6\columnwidth]{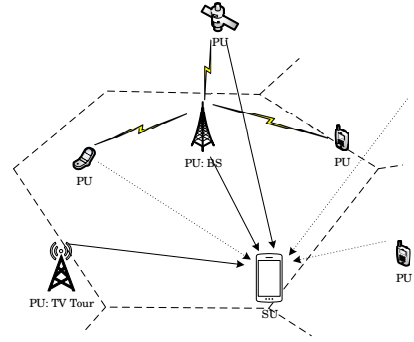}
\caption{A $SU$ performing wideband spectrum sensing. The received signals are coming from PUs with different levels of energy.}
\label{fig:sys_mod}}
\end{figure}
The vector $\boldsymbol{x}$ contains a faded version of the PUs' signals operating in the different bands. Given the occupancy of the bands by their \PU s (as illustrated in Fig.~\ref{fig:band_ocup}(b)) and in the absence of fading and interference, the vector $\boldsymbol{x}$ can be considered as {\em sparse}, where sparsity is formally defined as follows.
\begin{definition}
A vector $\boldsymbol{x}\in\mathbb{R}^n$ is \emph{k-sparse} if it has (with or without a basis change) at most $k$ non-zero elements~\cite{davenport2011introduction}; i.e., $supp(\boldsymbol{x})=\|\boldsymbol{x}\|_{\ell_0} = |\{i : x_i\neq 0\}|\leq k$. The set of $k-$sparse vectors in $\mathbb{R}^n$ are denoted by $\Sigma_{k}=\{\boldsymbol{x}\in\mathbb{R}^n : \|\boldsymbol{x}\|_{\ell_0}\leq k\}$.
\end{definition}
In practice, however, there will likely be interference coming from other nearby cells and users, and hence, $\boldsymbol{x}$ could rather be {\em nearly sparse}, formally defined as follows.
\begin{definition}
  A vector $\boldsymbol{x}\in\mathbb{R}^n$ is \emph{nearly sparse} (or also compressible~\cite{davenport2011introduction}) if most of its components obey a fast power law decay. The $k-$sparsity index of $\boldsymbol{x}$ is then defined as $\sigma_k(\boldsymbol{x},\|.\|_{\ell_p})=\displaystyle{\min_{\boldsymbol{z}\in\Sigma_k}}\|\boldsymbol{x}-\boldsymbol{z}\|_{\ell_p}$.
\end{definition}
Since wideband spectrum is large, the number of required samples can be huge, making the sensing operation prohibitively costly and the needed hardware capabilities beyond possible. To overcome this issue, compressive sampling theory has been relied on as a way to reduce the number needed measurements, given that wideband spectrum signals contain some sparsity or nearly sparsity property.
After performing the compressive sampling, the resulted signal can be written as
\begin{eqnarray}\nonumber
  \boldsymbol{y} &=& \Psi {\mathcal{F}^{-1}}(\boldsymbol{x} +\boldsymbol{w}_f) \\\nonumber
   &=& {\mathcal{A}}\boldsymbol{x}+\boldsymbol{\eta},
\end{eqnarray}
where $\boldsymbol{y}\in\mathbb{R}^m$ is the measurement vector, ${\mathcal{F}^{-1}}$ is the inverse discrete Fourier transform, and $\Psi$ is the sensing matrix assumed to have a full rank, i.e. $rank(\Psi)=m$. The sensing noise $\boldsymbol{\eta}$ is equal to $\Psi {\mathcal{F}^{-1}}\boldsymbol{w}_f$.
It is worth mentioning that from a practical viewpoint, wideband spectrum sensing requires: $i$) wideband antennas, $ii$) wideband front-end filters, and $iii$) high speed analog-to-digital converters (ADC), which are known to be very challenging to build~\cite{yoon2010ultra,hao2011highly,abari2013analog,abari2013analog,kirolos2006analog}.
Compressive sampling allows to overcome this issue by sampling at sub-Nyquist rate as illustrated by Fig.~\ref{fig:receiver_archi}.
The signal is first amplified by $m$ amplifiers and mixed with a pseudo-random waveform at a Nyquist rate ($f_s=2f_{\max}$). Then, an integrator is applied followed by an ADC that takes samples at sub-Nyquist rate ($f_s/n$).
\begin{figure}
\centering{
\includegraphics[width=1\columnwidth]{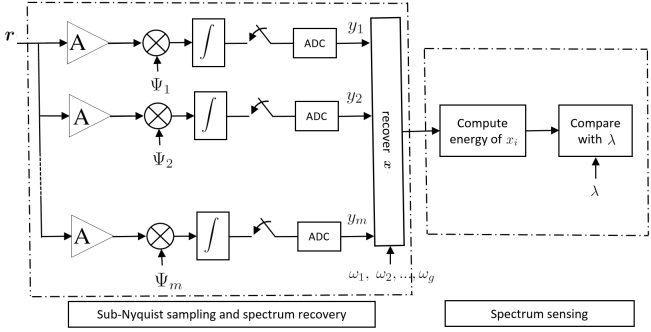}
\caption{Illustration of an \SU~receiver architecture.}
\label{fig:receiver_archi}}
\end{figure}

Different from the classical application of compressive sampling for wideband spectrum sensing, in this paper we propose to take advantage of the block-like structure of the occupancy of the wideband spectrum, and design an efficient compressive spectrum sensing algorithm well suited for heterogeneous wideband CRNs. Exploiting the variability of the average band occupancies across the various blocks has the potential for improving the recovery of the wideband spectrum sensing signals, and therefore, the ability of acquiring accurate \PU~detection and spectrum availability information efficiently.

\section{The Proposed Wideband Spectrum Sensing Information Recovery}
\label{sec:proposed}
The sensing matrix and recovery algorithm are the main challenging components in compressive sampling design. While the former consists of minimizing the number of measurements, the latter consists of ensuring a stable and robust recovery. In this work, we exploit the block-like occupancy structure information of the wideband spectrum to propose a new recovery algorithm that outperforms existing approaches by $1)$ requiring lesser numbers of measurements (better sensing matrix) and $2)$ reducing recovery error (more stable and robust recovery).
In this section, we start by providing some background on signal recovery using classical compressive sampling. Then, we present our proposed approach, and analyze its performance by bounding its achievable mean square errors and its required number of measurements.

\subsection{Background}
To acquire spectrum availability information, an SU needs first to recover the frequency-domain version of the received signal. Exploiting the fact that the signal is sparse, an ideal recovery can be performed by minimizing the $\ell_0-$norm of the signal. This, however, happens to be NP-hard~\cite{Candestao2005TIT}. It turns out that minimizing the $\ell_1-$norm recovers the sparsest solution with a bounded error that depends on the noise variance and the solution structure~\cite{candes2006stable}. This can be formulated as
\begin{equation*}
\begin{aligned}
\mathscr{P}_1 : &\; \underset{x}{\text{minimize}}
& & \|\boldsymbol{x}\|_{\ell_1}\\
& \text{subject to}
& & \|\mathcal{A}\boldsymbol{x}-\boldsymbol{y}\|_{\ell_2}\leq \epsilon
\end{aligned}
\end{equation*}
Here, $\epsilon$ is a user-defined parameter chosen such that $\|\boldsymbol{\eta}\|_{\ell_2}\leq \epsilon$. This formulation is known also as Least Absolute Shrinkage and Selection Operator (LASSO)~\cite{candes2006stable}.

Although LASSO is shown to achieve good performance when applied for wideband spectrum sensing recovery, it does not capture, nor exploit the block-like occupancy structure information that is inherent to the heterogeneous wideband spectrum, where the occupancy is homogeneous within each block but heterogenous across the different blocks of the spectrum.
As we will show later, it is the exploitation of this block-like spectrum occupancy structure that is behind the performance again achieved by our proposed compressive spectrum sensing recovery algorithm.

\subsection{The Proposed Recovery Algorithm}

Intuitively, our key idea consists of incorporating and exploiting the sparsity level variability across the different blocks of the spectrum sensing signal to perform intelligent solution search. We essentially encourage more search of the non-zero elements of the signal $\boldsymbol{x}$ in the blocks that have higher average sparsity levels while discouraging this search in the blocks with low average sparsity levels.
Such a variability in the block sparsity levels can be incorporated in the formulation through carefully designed weights. More specifically, we propose
the following weighted $\ell_1-$minimization recovery scheme:
\begin{equation*}
\begin{aligned}
\mathscr{P}_1^{\omega}: & & \underset{x}{\text{minimize}}
& & \sum_{l=1}^g \omega_l\|\boldsymbol{x}_{l}\|_{\ell_1}\\
& & \text{subject to}
& &  \|\mathcal{A}\boldsymbol{x}-\boldsymbol{y}\|_{\ell_2}\leq \epsilon.
\end{aligned}
\end{equation*}
where $\boldsymbol{x}=[\boldsymbol{x}_1^T,..., \boldsymbol{x}_g^T]^T$, $\boldsymbol{x}_l^T$ is a $n_l\times 1$ vector, and $\omega_l$ is the weight assigned to block $l$ for $l\in\{1,...,g\}$.

The question that arises here now is how to design and select these weights.
Intuitively, given that the average sparsity level differs from one block to another, blocks with higher average sparsity levels should contain more occupied bands than those blocks with lower averages. This means that if we consider two blocks with two different average sparsity levels, say $\bar{k}_1$ and $\bar{k}_2$, such that $\bar{k}_1<\bar{k}_2$, then to encourage the search for more occupied bands in the second block, the weight $\omega_2$ assigned to the second block should be smaller than the weight $\omega_1$ assigned to the first block. Following this intuition, we set the weights to be inversely proportional to the average sparsity levels. More specifically,
\begin{equation}\label{eqn:weight}
\omega_i=\frac{1/\bar{k}_i}{\sum_{j=1}^g1/\bar{k}_j}~~~~\forall~i \in \{1,...,g\}
\end{equation}

\begin{remark}\emph{Some insights into the proposed scheme}\\
Consider a two-block spectrum with $\bar{k}_1>\bar{k}_2$, and hence, with $\omega_2>\omega_1$. For this special case, the recovery algorithm can then be re-written as
\begin{equation*}
\begin{aligned}
\mathscr{P}_1^{\omega,2}: & & \underset{x}{\text{minimize}}
& & \|\boldsymbol{x}\|_{\ell_1}+(\frac{\omega_2}{\omega_1}-1)\|\boldsymbol{x}_{2}\|_{\ell_1}\\
& & \text{subject to}
& &  \|\mathcal{A}\boldsymbol{x}-\boldsymbol{y}\|_{\ell_2}\leq \epsilon.
\end{aligned}
\end{equation*}
Since we are minimizing the $\ell_1-$norm of $\boldsymbol{x}$ and the $\ell_1-$norm of $\boldsymbol{x}_2$, this can be interpreted as ensuring that the vector $\boldsymbol{x}$ is sparse while ensuring that the portion $\boldsymbol{x}_{2}$ of $\boldsymbol{x}$ is also sparse (since $\frac{\omega_2}{\omega_1}-1>0$). This means that all solutions that are sparse as a whole but somehow dense in their second portion are eliminated.
\end{remark}

\begin{remark}{\em Weights design}\\
The proposed scheme relies on the average occupancy of blocks at a per-block granularity to be able to improve the recovery accuracy of sensed signal.
From a practical viewpoint, one approach of acquiring the average occupancy (sparsity level) of each block is by monitoring the occupancy of each band within the block and averaging them over time, as already been proposed in~\cite{yilmaz2016determination,mehdawi2015spectrum}. Other machine learning based prediction approaches can also be used to provide good estimates of the average occupancy.
That is said, we also want to mention that even when the average occupancy is not determined on a per-block basis; i.e., the entire wideband spectrum is considered as one block, our proposed algorithm becomes equivalent to the classical $\ell_1$-minimization approach (LASSO) (i.e., $\mathscr{P}_1$). In other words, our algorithm performs similarly to LASSO when average block occupancices are unavailable and outperforms it otherwise.
\end{remark}
In the remaining of this section, we derive and evaluate the performance achievable by the proposed recovery algorithm by showing that it $1)$ incurs errors smaller than those incurred by existing techniques and $2)$ reduces the sensing overhead by requiring smaller numbers of required measurements.

\subsection{Mean Square Error Analysis}
The following theorem shows that our weighted recovery algorithm incurs lesser errors than what LASSO~\cite{candes2006stable} incurs.

\begin{theorem}\label{theo:error}
Letting $\boldsymbol{x}^{\sharp}$ be the optimal solution for $\mathscr{P}_1^{\omega}$, $\boldsymbol{x}^{\dag}$ the optimal solution for $\mathscr{P}_1$ and $\boldsymbol{y}=\mathcal{A}\boldsymbol{x}_0+\boldsymbol{\eta}$,
we have
\begin{displaymath}
\label{eqn:perf}
\|\boldsymbol{x}^{\sharp}-\boldsymbol{x}_0\|_{\ell_2} \leq \|\boldsymbol{x}^{\dag}-\boldsymbol{x}_0\|_{\ell_2}.
\end{displaymath}
with a probability exceeding
\begin{eqnarray}\label{eqn:vio}\nonumber
1-\sum_{i=1}^{g-1}\sum_{j=i+1}^{g}\sum_{k=1}^{\min(n_i,n_j)}\sum_{l=0}^{k-1}\dbinom{n_i}{l}q_i^l(1-q_i)^{n_i-l}\\
\times\dbinom{n_j}{k} q_{j}^k(1-q_j)^{n_j-k}
\end{eqnarray}
\end{theorem}
assuming $n_1q_1\geq...\geq n_gq_g$.
\begin{proof}
The proof is provided in Appendix~\ref{idx:prooftheoerror}.
\end{proof}
The theorem says that the solution to the proposed $\mathscr{P}_1^{\omega}$ is at least as good as the solution to $\mathscr{P}_1$. Also as done by design, the more heterogeneous the wideband spectrum is, the higher the error gap between our proposed algorithm and LASSO is.
This is because the searched solution has the right required structure captured via the assigned weights.

Now, we assess the stability and robustness of the proposed recovery scheme, defined as follows.
\begin{definition}{\em Stable and Robust Recovery~\cite{candes2006stable}}\\
For $\boldsymbol{y}=\mathcal{A}\boldsymbol{x}+\boldsymbol{w}$ such that $\|\boldsymbol{w}\|_{\ell_2}\leq \epsilon$, a recovery algorithm, $\Delta$, and a sensing matrix, $\mathcal{A}$, are said to achieve a stable and robust recovery if there exist $C_0$ and $C_1$ such that
  \begin{equation}\label{eqn:recover}\nonumber
  \|\Delta \boldsymbol{y}-\boldsymbol{x}\|_{\ell_2}\leq C_0 \epsilon+C_1 \frac{\sigma_k(\boldsymbol{x},\|.\|_{\ell_p})}{\sqrt{k}}.
  \end{equation}
\end{definition}
Note that the stability implies that small perturbations of the observation lead to a small perturbation of the recovered signal. Robustness, on the other hand, is relative to noise; for instance, if the measurement vector is corrupted by noise with a bounded energy, then the error is also bounded~\cite{candes2006stable}.
We now state the following result, which follows directly from Theorem~\ref{eqn:perf}.

\begin{proposition}\label{prop:1}
Our proposed algorithm, $\mathscr{P}_1^{\omega}$, achieves a stable and robust recovery.
\end{proposition}
\begin{proof}
The proof is provided in Appendix~\ref{idx:proofprop1}.
\end{proof}

The proposition gives a bound on the incurred error by means of two quantities. The first is an error of the order of the noise variance while the second is of the order of the sparsity index of $\boldsymbol{x}$.

\begin{remark}\emph{Effect of time-variability}\\
We want to iterate that our proposed algorithm is guaranteed to outperform existing approaches on the average, and not on a per-sensing step basis.
This is because although the performance improvement achieved by our technique stems from the fact that blocks with higher average sparsity levels are given lower weights---which is true on the average, it is not unlikely that, at some sensing step, the actual sparsity level of a block with a higher average could be smaller than that of a block with a lower average. When this happens, our algorithm won't be guaranteed to achieve the best performance during that specific sensing step. The good news is that first what matters is the average over longer periods of sensing time, and second,  depending on the gap between the block sparsity averages, this scenario happens with very low probability.

To illustrate, let us assume that the wideband spectrum contains two blocks with average sparsity  $\bar{k}_1=\sum_{j\in \mathcal{G}_1}p_j\approx n_1p_1$ and $\bar{k}_2=\sum_{j\in \mathcal{G}_2}p_j\approx n_2p_2$ with $\bar{k}_2<\bar{k}_1$, where again $|\mathcal{G}_1|=n_1$ and $|\mathcal{G}_2|=n_2$. Here, the occupancy probabilities of all bands in each of these two blocks are assumed to be close to one another. Our approach encourages to find more occupied bands in the first block than in the second block. However, since band occupancy is time varying, then at some given time we may have a lesser number of non-zero components in first block than in the second. This unlikely event, in this scenario, happens with probability
\begin{displaymath}
\sum_{k=1}^{\min(n_1,n_2)}\sum_{l=0}^{k-1}\dbinom{n_1}{l}q_1^l(1-q_1)^{n_1-l}\dbinom{n_2}{k} q_{2}^k(1-q_2)^{n_2-k}
\end{displaymath}
For a sufficiently different average sparsity levels (e.g. having $\bar{k}_1>2\bar{k}_2$), this probability is smaller than $0.02$.
Finally, it is worth mentioning that our proposed scheme can achieve further performance improvement by adopting advanced estimation approaches, such as those that are based on machine learning~\cite{wang2012sparsity}. However, this additional performance improvement comes at the price of additional computational complexity that is accompanied with these estimators.
\end{remark}

Having investigated the design of the recovery algorithm, now we turn our attention to the design of the sensing matrix. The number of measurements, $m$, that needs to be taken determines the size of the sensing matrix and hence the sensing overhead of the recovery approach. Therefore, we aim to exploit the structure of the solution to reduce the required number of measurements as much as possible, so that the sensing overhead is reduced as much as possible.

\subsection{Number of Required Measurements}
The sensing matrix is usually designed with two major design criteria/goals in mind: reducing the number of measurements and satisfying the RIP property, defined as follows.
\begin{definition}\label{def:1}\emph{Restricted Isometry Property (RIP)}~\cite{davenport2011introduction}\\
A matrix $\mathcal{A}$ is said to satisfy the RIP of order $k$ if there exists $\delta_{k}\in(0,1)$ such that for $\boldsymbol{x}\in\Sigma_{k}$
\begin{equation}\label{eqn:RIP}\nonumber
(1-\delta_k)\|\boldsymbol{x}\|_{\ell_2}^2\leq \|\mathcal{A}\boldsymbol{x}\|_{\ell_2}^2\leq (1+\delta_k)\|\boldsymbol{x}\|_{\ell_2}^2.
\end{equation}
\end{definition}
Broadly speaking, the RIP ensures that every $k$ columns of $\mathcal{A}$ are nearly orthogonal.
We now present one of our main results derived in this paper, which provides a lower bound on the number of required measurements.
\begin{theorem}\label{theo:mea}
Let $\mathcal{A}=[\mathcal{A}_1...\mathcal{A}_g]$ be the sensing matrix such that $\mathcal{A}_i$ satisfies the RIP of order $2\bar{k}_i$ with $\{\delta_{2\bar{k}_1},...,\delta_{2\bar{k}_g}\}\in(0,1/2]$. Then, the number of measurements $m$ must satisfy
\begin{equation}\label{eqn:m}\nonumber
  m\geq \frac{1}{2\log\Big(\frac{\sum_{i=1}^g \sqrt{2\bar{k}_i(1+\delta_{\bar{k}_i})}+\max_i(\sqrt{\bar{k}_i(1-\delta_{\bar{k}_i})/8})}{\min_i(\sqrt{\bar{k}_i(1-\delta_{\bar{k}_i})/8})}\Big)}\bar{k}\log\Big(\frac{n}{\bar{k}}\Big)
\end{equation}
\end{theorem}
\begin{proof}
The proof is provided in Appendix~\ref{idx:prooftheomea}.
\end{proof}
Theorem~\ref{theo:mea} given above provides a lower bound on the required number of measurements needed to recover the signal. As shown later in the result section, this bound is tighter than existing approaches in that with the same number of measurements, our proposed framework can recover signals with better accuracy than those obtained via existing approaches. Alternatively, we can also say that our framework can recover signals with an accuracy equal to those obtained with existing approaches, but while requiring lesser numbers of measurements, $m$.
The derived lower bound exhibits an asymptotic behavior similar to that of the classic bound (i.e., $\mathcal{O}(\bar{k} \log(n/\bar{k}))$), but with a smaller constant. By setting $g=1$, we get the bound provided in~\cite[Theorem 1.4]{davenport2011introduction}. So our derived bound could be viewed as a generalization of that of~\cite{davenport2011introduction}, in that it
applies to wideband spectrum with heterogeneous block occupancies; setting $g=1$ corresponds to the special case of the homogeneous wideband spectrum.

Existing approaches determine the required number of measurements by setting the sparsity level to the average number of occupied bands (e.g., $m\geq \bar{k} \log(n/\bar{k})$). However, in wideband spectrum sensing, the number of occupied bands changes over time, and can easily exceed the average number. Every time this happens, it leads to an inaccurate signal recovery (it yields a solution with high error).
To address this issue, in our proposed framework, we do not base the selection of the number of measurements on the average sparsity. Instead, the sparsity level that we use in Theorem~\ref{theo:mea} to determine $m$ is chosen in such a way that the likelihood that the number of occupied bands exceeds that number is small. The analysis needed to help us determine such a sparsity level is provided in the next section.

\subsection{PU Traffic Characterization}
Based on the model of occupancy of the wideband provided in the system model, the following lemma gives the probability mass distribution of the number of occupied bands.
\begin{lemma}
The number of occupied bands across the entire wideband has the following probability mass function
\begin{equation}\label{eqn:bingen}\nonumber
  \textrm{Pr}(X=k) = \displaystyle{\sum_{\Lambda\in\mathcal{S}_k}}\Big[ \displaystyle{\prod_{i\in\Lambda}}~p_i\Big]\Big[     \displaystyle{\prod_{j\in\Lambda^c}}(1-p_j)\Big]
\end{equation}
where $\mathcal{S}_k=\{\Lambda:~\Lambda\subseteq \{1,...,n\}, |\Lambda|=k\}$, and $\Lambda^c$ is the complementary set of $\Lambda$.
\end{lemma}
\begin{proof}
Let $\Lambda$ the support such that its $i^{th}$ component is equal to one when there is a PU using the $i^{th}$ band. Then, the probability that there is exactly $k$ occupied bands is $\Big[  \displaystyle{\prod_{i\in\Lambda}}~p_i\Big]\Big[\displaystyle{\prod_{j\in\Lambda^c}}(1-p_j)\Big]$ such that $| \Lambda|=k$. Now, considering all the supports with a cardinality $k$ gives the expression of the mass distribution.
\end{proof}

Given this distribution, the average number of occupied bands across the entire wideband spectrum is $\bar{p}=\sum_{i=1}^np_i$. As just mentioned earlier, setting the sparsity level to be fixed to the average $\lfloor\bar{p}\rfloor$ will lead to inaccurate signal recovery, since the likelihood that the number of occupied bands exceeds this sparsity level is not negligible.
In the following theorem, we provide a lower bound on the probability that the number of occupied bands is below an arbitrary sparsity level.

\begin{theorem}\label{theo:bound}
  The probability that the number of occupied bands is below a sparsity level $k_0$ is lower-bounded by
  \begin{eqnarray}\label{eqn:bound}\nonumber
    \textrm{Pr}(X\leq k_0) &=&\sum_{k=0}^{k_0}\displaystyle{\sum_{\Lambda\in\mathcal{S}_k}}\Big[     \displaystyle{\prod_{i\in\Lambda}}~p_i\Big]\Big[     \displaystyle{\prod_{j\in\Lambda^c}}(1-p_j)\Big]\\
    &\geq& 1-\frac{e^{k_0-\sum_{i}^np_i}}{(k_0/\sum_{i}^np_i)^{k_0}}
  \end{eqnarray}
\end{theorem}
\begin{proof}
The proof is provided in Appendix~\ref{idx:prooftheobound}.
\end{proof}
Since the sparsity level is a time-varying process, this theorem gives a probabilistic bound on how to choose a sparsity level such that the level will be exceeded only with a certain probability. Now depending on the allowed fraction, $\alpha$, of instances in which the actual number of occupied bands exceeds the sparsity level, Theorem~\ref{theo:bound} can be used to determine the sparsity level, $k_0$, that can be used in Theorem~\ref{theo:mea} to determine the required number of measurements, $m$.
In other words, $\alpha$ is the probability that the actual number of occupied bands is above the defined sparsity level $k_0$.
If $\alpha$ is set to $5\%$, then it means that only about $5\%$ of the time the actual number of occupied bands exceeds $k_0$. As expected, there is a clear tradeoff between $\alpha$ and $k_0$. Smaller values of $\alpha$ requires higher values of $k_0$, and vice-versa. In our numerical evaluations given in the next section, $\alpha$ is set to $4\%$.

\section{Numerical Evaluation}
\label{sec:numerical_results}
In this section, we evaluate our proposed wideband spectrum sensing approach and we compare its performance to the state-of-the-art approaches.
Consider a primary system operating over a wideband consisting of $n=256$ bands.
We assume that the wideband contains $g=4$ blocks with equal sizes. The average probabilities of occupancy in each block are as follows: $\bar{k}_1=0.1\times 64$, $\bar{k}_2=0.01\times 64$, $\bar{k}_3=0.1\times 64$, $\bar{k}_4=0.01\times 64$.
To model the signals coming from the active users, we generate them in the frequency domain with random magnitudes (which captures the effect of the different channel SNRs that every operating PU has with the SU).
At the SU side, the sensing matrix $\Psi$ is generated according to a Bernoulli distribution with zero mean and $1/m$ variance. We opted for a sub-Gaussian distribution since it guarantees the RIP with high probability~\cite{davenport2011introduction}. Here, the number of measurements is generated first according to $m=\mathcal{O}(k_0\log(n/k_0))$.

We fix $k_0$ to $25$ which according to Theorem~\ref{theo:bound} is satisfied with a probability that exceeds $0.96$ (see Fig. \ref{fig:sparsity}). Now assuming an RIP constant $\delta_{2k_i}\leq1/2$ and replacing $k_0$ and the RIP constant with their values in Theorem 3 yields that the number of measurements should be at least $29$.
\begin{figure}
\centering{
\includegraphics[width=1\columnwidth]{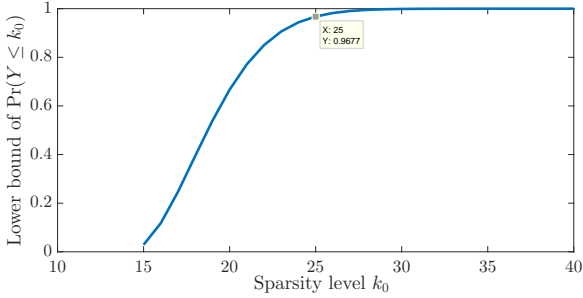}
\caption{Lower bound of $\textrm{Pr}(X<k_0)$ as a function of the sparsity level $k_0$. }
\label{fig:sparsity}}
\end{figure}
We use CVX for the solving of the optimization problem~\cite{grant2008cvx}.

A first performance that we look at is the mean square error $\|\boldsymbol{x}^{\sharp}-\boldsymbol{x}_0\|_{\ell_2}$ as a function of the sensing SNR defined as $\textrm{SNR}=\frac{\|\mathcal{A}\boldsymbol{x}\|_{\ell_2}^2}{\|\boldsymbol{\eta}\|_{\ell_2}^2}$,
where $\|\mathcal{A}\boldsymbol{x}\|_{\ell_2}^2=(\mathcal{A}\boldsymbol{x})^T\mathcal{A}\boldsymbol{x}$ and $\|\boldsymbol{\eta}\|_{\ell_2}^2=\boldsymbol{\eta}^T\boldsymbol{\eta}$.
In Fig.~\ref{fig:perf1}, we compare our proposed technique to the existing approaches. Compared to LASSO~\cite{candes2006stable}, CoSaMP~\cite{needell2009cosamp}, and (OMP)~\cite{tropp2007signal}, our proposed approach achieves a lesser error when fixing the number of measurement $m$ to $27$. This is because we account for the average sparsity levels in each block, thereby favoring the search on the first and third block rather than the two others. Also, observe that as the sensing SNR gets better, not only does the error of the proposed technique decrease, but also the error gap between our technique and that of the other ones increases. This is because the noise effect becomes limited. Furthermore, OMP has the worst performance as it requires a higher number of measurements to perform well. In Fig.~\ref{fig:perfSNR}, we look at the performance of the recovery scheme as a function of the average received SNR defined as the ratio between the received signal power and the noise power; i.e., $\|\boldsymbol{x}\|_{\ell_0}^2/\|\boldsymbol{\eta}\|_{\ell_2}^2$. We observe a similar behavior as in Fig.~\ref{fig:perf1}.

\begin{figure}
\centering{
\includegraphics[width=1\columnwidth]{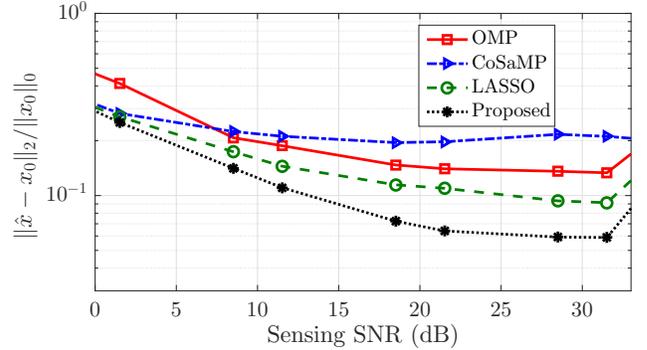}
\caption{Comparison between the recovery approaches in terms of mean square error as a function of the sensing SNR ($m=27$).}
\label{fig:perf1}}
\end{figure}
\begin{figure}
\centering{
\includegraphics[width=1\columnwidth]{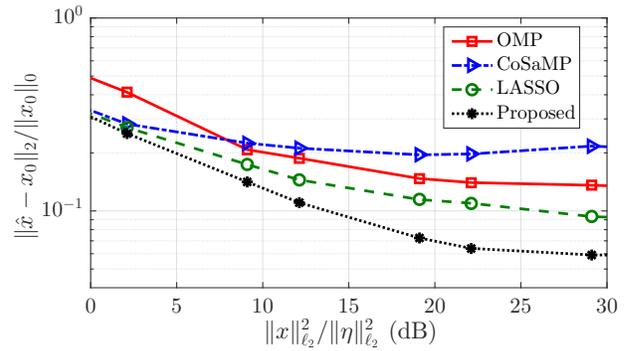}
\caption{Comparison between the recovery approaches in terms of mean square error as a function of received signal SNR ($m=27$).}
\label{fig:perfSNR}}
\end{figure}

In Fig.~\ref{fig:perf2}, we investigate the error percentage gain (EPG) achieved by our technique when compared to the other schemes under various different numbers of measurements. We define the error gain of our approach over an existing approach $i$ as
\begin{equation}\nonumber
\text{EPG} (\%)=\frac{\|\boldsymbol{x}_{i}^{\sharp}-\boldsymbol{x}_0\|_{\ell_2}-\|\boldsymbol{x}_{\textrm{Proposed}}^{\sharp}-\boldsymbol{x}_0\|_{\ell_2}}{\|\boldsymbol{x}_{i}^{\sharp}-\boldsymbol{x}_0\|_{\ell_2}}100\%
\end{equation}
 \begin{figure}
\centering{
\includegraphics[width=1\columnwidth]{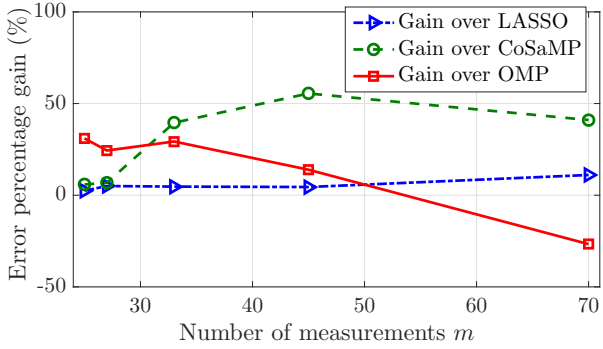}
\caption{Error gain comparison with LASSO~\cite{candes2006stable}, CoSaMP~\cite{needell2009cosamp}, and (OMP)~\cite{tropp2007signal} for SNR$=20dB$.}
\label{fig:perf2}}
\end{figure}
Observe that when the number of measurements is low, our proposed technique outperforms the other three techniques. But when the number of measurements $m$ is relatively high, our technique still performs better than CoSaMP and LASSO, but worse than OMP. However, OMP achieves this superior performance only under high number of measurements, a range that is not of interest due to its high incurred overhead.

After recovering the signal and in order to decide on the availability of the different bands, we compare the energy of the recovered signal in every band with the threshold~\cite{digham2007energy}, $ \lambda=\frac{\mathbb{E}(\|\boldsymbol{\eta}\|_{\ell_2}^2)}{m}\Big(1+\frac{Q^{-1}(P_f)}{\sqrt{1/2}}\Big)$,
where $P_f$ is a user-defined threshold for the false alarm probability. It is defined as the probability that a vacant band is detected as occupied, and is expressed as $\frac{1}{\sum_{i=1}^n\mathcal(1-{H}_i)}\sum_{i=1}^nPr(|x_i|^2
\geq \lambda|\mathcal{H}_i=0)$.
${Q^{-1}}$ is the inverse of the $Q-$function. In Fig.~\ref{fig:detec}, we plot this detection probability as a function of the false probability for a fixed average sensing SNR, where the detection probability is computed as $\frac{1}{\sum_{i=1}^n\mathcal{H}_i}\sum_{i=1}^nPr(|x_i|^2\geq \lambda |\mathcal{H}_i=1)$.
 \begin{figure}
\centering{
\includegraphics[width=1\columnwidth]{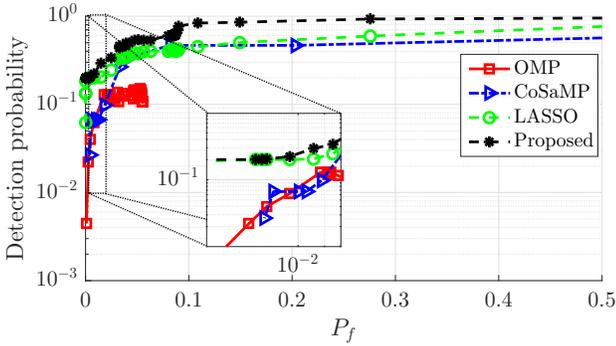}
\caption{Probability of detection as a function of the probability of false alarm with number of measurements $m=27$ and sensing SNR$=33~dB$.}
\label{fig:detec}}
\end{figure}
Although the number of measurements is less than what is required, our proposed technique has the best detection capability among all other approaches. This also confirms the result of Fig.~\ref{fig:perf2}.

\section{Conclusion}
\label{sec:conclusion}
We proposed an efficient wideband spectrum sensing technique based on compressive sampling. Our proposed technique is a weighted $\ell_1-$minimization recovery approach that accounts for the block-like structure inherent to the heterogeneous  nature of wideband spectrum allocation. We showed that the proposed approach outperforms existing approaches by achieving lower mean square errors, enabling higher detection probability, and requiring lesser numbers of measurements when compared to the-state-of-the-art approaches.

\appendices

\section{Proof of Theorem \ref{theo:error}}\label{idx:prooftheoerror}
Let us consider the average sparsity level in every block to be $\bar{k}_i=p_i.n_i$ and define the weights as $\omega_i=\frac{1}{\bar{k}_i}$ (and then we normalize it, as in Equation~\eqref{eqn:weight}, as $\omega_i=\omega_i/\sum_{j=1}^n\omega_j$). Without loss of generality, we assume that $\omega_1\leq\omega_2\leq ...\leq\omega_g$.
First, let us assume to have only knowledge of $\bar{k}_1$ to have the highest sparsity level in all the blocks. Then, we can consider the recovery problem as
\begin{equation*}
\begin{aligned}
\mathscr{P}_1^{ \omega_1,1}: \underset{x}{\text{minimize}} & & \omega_1\|\boldsymbol{x}_{1}\|_{\ell_1}+ \sum_{l=2}^g\|\boldsymbol{x}_{l}\|_{\ell_1}\\
\text{subject to}& &\|\mathcal{A}\boldsymbol{x}-\boldsymbol{y}\|_{\ell_2}\leq \epsilon.
\end{aligned}
\end{equation*}
Since we have $\omega_{1}\leq 1$, this means we encourage the search of more components of $\boldsymbol{x}$ in the first than in the second block.
We know that the set of solutions are given by $\boldsymbol{x}_0+\mathcal{N}ull(\mathcal{A})$. Ideally, its intersection with the $\ell_1-$ball gives the minimizer of $\mathscr{P}_1$. Now by introducing the weight in the first block, the weighted norm ball will be pinched towards the axis containing $\boldsymbol{x}_1$ which has, in average, lot of non-zero components. Therefore, the recovered vector from $\mathscr{P}_1^{ \omega_1,1}$ is going to be more accurate than the recovered vector from $\mathscr{P}_1$.

Now, assume to have the knowledge of $1\leq i < g$ sparsity level of $i$ blocks. Then, the optimization can be written as
\begin{equation*}
\begin{aligned}
\mathscr{P}_1^{ \omega_1, \omega_2,..., \omega_i,1}: \underset{x}{\text{minimize}} & & \sum_{l=1}^i\omega_l\|\boldsymbol{x}_{l}\|_{\ell_1}+ \sum_{l=i+1}^g \|\boldsymbol{x}_{l}\|_{\ell_1}\\
\text{subject to}& &\|\mathcal{A}\boldsymbol{x}-\boldsymbol{y}\|_{\ell_2}\leq \epsilon.
\end{aligned}
\end{equation*}
Applying the same observation, the weighted norm ball is pinched more towards the components of the denser blocks. Therefore, the performance should be at least the performance of $\mathscr{P}_1$.
Setting $l=g$, we get $\|\boldsymbol{x}^{\sharp}-\boldsymbol{x}_0\|_{\ell_2}\leq \|\boldsymbol{x}^{\dag}-\boldsymbol{x}_0\|_{\ell_2}$.
On the other hand, the bands' occupation is a random process following the bernoulli, then at some given time we may have a lesser number of non-zero components in the $i^{th}$ block than in the $j^{th}$ block with $(j>i)$, the event can be quantified as
\begin{eqnarray}\nonumber
\sum_{k=1}^{\min(n_i,n_j)}\sum_{l=0}^{k-1}\dbinom{n_i}{l}q_i^l(1-q_i)^{n_i-l}\dbinom{n_j}{k} q_{j}^k(1-q_j)^{n_j-k}.
\end{eqnarray}
Examining all the cases and taking the complementary, we get Equation~\eqref{eqn:vio}.

\section{Proof of Proposition~\ref{prop:1}}\label{idx:proofprop1}

Our proposed approach achieves a stable and robust recovery if we can find $C_0$ and $C_1$ such that
  \begin{equation}\nonumber
 \|\boldsymbol{x}^{\sharp}-\boldsymbol{x}_0\|_{\ell_2}\leq C_0 \epsilon+C_1 \frac{\sigma_k(\boldsymbol{x},\|.\|_{\ell_p})}{\sqrt{k}}.
  \end{equation}
Combining Theorem~\ref{theo:error} and ~\cite[Theorem~2]{candes2006stable}, we get (with a probability exceeding \eqref{eqn:vio})
\begin{eqnarray}\nonumber
\|\boldsymbol{x}^{\sharp}-\boldsymbol{x}_0\|_{\ell_2}&\leq&\|\boldsymbol{x}^{\dag}-\boldsymbol{x}_0\|_{\ell_2}\\ \nonumber
&\leq & C_0.\epsilon+C_1. \frac{\sigma_k(\boldsymbol{x}_0,\|.\|_{\ell_1})}{\sqrt{k}}
\end{eqnarray}
where
\begin{equation}\label{eqn:c0}
C_0=\frac{2(1+1/\sqrt{a})}{\sqrt{1-\delta_{(a+1)k}}-\sqrt{1+\delta_{ak}}/\sqrt{a}}
\end{equation}
and \begin{equation}\label{eqn:c1}
C_1=\frac{2\sqrt{1-\delta_{(a+1)k}}+\sqrt{1+\delta_{ak}}/\sqrt{a}}{\sqrt{a}\sqrt{1-\delta_{(a+1)k}}-\sqrt{1+\delta_{ak}}}
\end{equation}
with $a$ and $b$ such that $\delta_{ak}+a\delta_{(a+1)k}<a-1$.
Therefore, our approach is stable and robust.
\section{Proof of Theorem~\ref{theo:mea}}\label{idx:prooftheomea}
Prior to give the proof of the theorem, we start by providing the following lemma.

\begin{lemma}\label{lemma:meas}
Let $\bar{k}=\sum_{i=1}^g\bar{k}_i$ and $n=\sum_{i=1}^gn_i$ with $\bar{k}_i\leq n_i/2$. There exists a set $X=\bigcup_{i=1}^gX_i\subset\Sigma_{\bar{k}}$ such that for any $x=[x_1^T x_2^T ...x_g^T]$ with $x_i\in X_i$ for $i=1,\ldots,g$, we have:~\\
(1) $\|x_i\|_{\ell_2}\leq \sqrt{\bar{k}_i}$~\\
(2) for any $x,y\in X$ with $x\neq y$, $\|x_i-y_i\|_{\ell_2}\geq \sqrt{\bar{k}_i/2}$ and $\log|X|\geq \frac{\bar{k}}{\bar{2}}\log\Big(\frac{n}{\bar{k}}\Big)$.

\end{lemma}
\begin{proof}
The proof of the lemma is similar to~\cite[Lemma A.1]{davenport2011introduction}. It is omitted here for brevity.
\end{proof}

The proof of the theorem is inspired from the proof in~\cite{davenport2011introduction} and based on Lemma~\ref{lemma:meas}.
First, we have $x=\sum_{i=1}^gx_i$ with $\|x_{i}\|_{\ell_0}\leq \bar{k}_i$.
Then, for any $x_i$ and $ y_i\in \Sigma_{2\bar{k}_i}$, we have according to the RIP property
\begin{eqnarray}\nonumber
\sqrt{1-\delta_{\bar{k}_i}}\|x_i-y_i\|_{\ell_2}&\leq& \|\mathcal{A}_ix_i-\mathcal{A}_iy_i\|_{\ell_2}\\ \nonumber
\|\mathcal{A}_ix_i-\mathcal{A}_iy_i\|_{\ell_2}& \leq& \sqrt{1+\delta_{\bar{k}_i}}\|x_i-y_i\|_{\ell_2}
\end{eqnarray}
Combining the above property with Lemma~\ref{lemma:meas}, we get
\begin{equation}\nonumber
\sqrt{\bar{k}_i(1-\delta_{\bar{k}_i})/2}\leq \|\mathcal{A}_ix_i-\mathcal{A}_iy_i\|_{\ell_2} \leq \sqrt{2\bar{k}_i(1+\delta_{\bar{k}_i})}.
\end{equation}
By considering the balls with radius $\tau_i$ such that $\tau_i=\sqrt{\bar{k}_i(1-\delta_{\bar{k}_i})/2}/2=\sqrt{\bar{k}_i(1-\delta_{\bar{k}_i})/8}$ centered at $\mathcal{A}_ix_i$, then these balls are disjoint. On the other hand, we have for any $x$ and $y \in \Sigma_{\bar{k}}$,
\begin{equation}\nonumber
\|\mathcal{A}x-\mathcal{A}y\|_{\ell_2} \leq \sum_{i=1}^g\|\mathcal{A}_ix_i-\mathcal{A}_iy_i\|_{\ell_2}\leq \sum_{i=1}^g \sqrt{2\bar{k}_i(1+\delta_{\bar{k}_i})}
\end{equation}
The upper bound gives an idea about the maximum distance between the centers of any pair of balls which is $d^{\max}=\sum_{i=1}^g \sqrt{2\bar{k}_i(1+\delta_{\bar{k}_i})}$. Therefore, all the balls are contained in the ball of radius $\tau=d^{\max}+\max_i(\tau_i)$.
Thus, we have
\begin{eqnarray}\nonumber
\textrm{Vol}\Big(B^m(\tau)\Big)\geq |X|\textrm{Vol}\Big(B^m(\min_i\tau_i)\Big),
\end{eqnarray}
where $\textrm{Vol}(B^m(\tau))$ is the volume of the ball which is given by $\textrm{Vol}(B^m(\tau))=\frac{\pi^{m/2}}{\Gamma (m/2+1)}{\tau}^m$ and $\Gamma(.)$ is the Euler Gamma function.
This yields
\begin{eqnarray}\nonumber
\Big(\frac{d^{\max}+\max_i(\tau_i)}{\min_i\tau_i}\Big)^m\geq |X|
\end{eqnarray}
Therefore, after applying $\log$, we get
\begin{eqnarray}\nonumber
m\geq \frac{1}{\log\Big(\frac{d^{\max}+\max_i(\tau_i)}{\min_i\tau_i}\Big)}\log(|X|)
\end{eqnarray}
Now recalling Lemma~\ref{lemma:meas}, we get $m\geq C_{\delta_{\bar{k}_1},...,\delta_{\bar{k}_{g}}}\bar{k}\log(n/\bar{k})$
where
\begin{eqnarray}\nonumber
C_{\delta_{\bar{k}_1},...,\delta_{\bar{k}_{g}}}
=&\frac{1}{2\log\Big(\frac{\sum_{i=1}^g \sqrt{2\bar{k}_i(1+\delta_{\bar{k}_i})}+\max_i(\sqrt{\bar{k}_i(1-\delta_{\bar{k}_i})/8})}{\min_i(\sqrt{\bar{k}_i(1-\delta_{\bar{k}_i})/8})}\Big)}.
\end{eqnarray}
which ends the proof.

\section{Proof of Theorem~\ref{theo:bound}}\label{idx:prooftheobound}
Let $Y=\sum_{i=1}^n{\mathcal{H}_i}$ be the random variable that contains the number of occupied bands. Since the occupation of the band is independent, then the moment generating function of $Y$ is given by
\begin{equation}\nonumber
\mathcal{M}_{Y}(t)=\prod_{i=1}^n(e^tp_i+1-p_i).
\end{equation}
 Now using the Chernoff bound, we have
 \begin{eqnarray}\nonumber
\textrm{Pr}(Y\geq k_0)&\leq& \inf_{t\geq 0}\Big\{e^{-k_0t} \mathcal{M}_{Y}(t)\Big\}\\ \nonumber
&=& \inf_{t\geq 0}\Big\{e^{-k_0t} \prod_{i=1}^n\big((e^t-1)p_i+1\big)\Big\}
 \end{eqnarray}
Using the fact that $e^x\geq 1+x$, we get
 \begin{eqnarray}\label{eqn:bnd}\nonumber
\textrm{Pr}(Y\geq k_0)&\leq& \inf_{t\geq 0}\Big\{e^{-k_0t} \prod_{i=1}^ne^{(e^t-1)p_i}\Big\}\\ \nonumber
&=& \inf_{t\geq 0}\Big\{e^{-k_0t} e^{(e^t-1)\sum_{i=1}^np_i}\Big\}\\\nonumber
&=& \inf_{t\geq 0}\Big\{\Big[ \underbrace{e^{(e^t-1)}e^{-t k_0/\sum_{i=1}^np_i}}_{(*)}\Big]^{\sum_{i=1}^np_i}\Big\}\nonumber
 \end{eqnarray}
To optimize $(*)$, we take the derivative over $t$ which yields to $t^*=\log(k_0/\sum_{i=1}^np_i)$. Now substituting $t^*$, we get
\begin{eqnarray}\nonumber
\textrm{Pr}(Y\geq k_0)&\leq& \frac{e^{k_0-\sum_{i=1}^np_i}}{(k_0/\sum_{i=1}^np_i)^{k_0}}
 \end{eqnarray}
Now since $\textrm{Pr}(Y\geq k_0)=1-\textrm{Pr}(Y\leq k_0)$, we get
\begin{eqnarray}\nonumber
1-\textrm{Pr}(Y\leq k_0)&\leq& \frac{e^{k_0-\sum_{i=1}^np_i}}{(k_0/\sum_{i=1}^np_i)^{k_0}}
 \end{eqnarray}
which gives the result of the theorem.

\bibliographystyle{IEEEtran}
\bibliography{References}
\end{document}